\documentclass[submission,copyright,creativecommons]{eptcs}
\usepackage[utf8]{inputenc}
\usepackage[english]{babel}
\usepackage{amsmath,amsfonts,amsthm,amssymb}
\usepackage{stmaryrd}
\usepackage[all]{xy}
\xymatrixrowsep{2ex}
\xymatrixcolsep{2ex}
\hypersetup{
  pdftitle={Presenting Finite Posets},
  pdfauthor={Samuel Mimram},
  unicode=true,
  colorlinks=true,
  linkcolor=black,
  citecolor=black,
  urlcolor=black
}
\let\P\undef
\let\C\undef
\usepackage{macros}
\usepackage{graphicx}

\newtheorem{theorem}{Theorem}
\newtheorem{proposition}[theorem]{Proposition}
\newtheorem{lemma}[theorem]{Lemma}

\theoremstyle{definition}
\newtheorem{definition}[theorem]{Definition}
\theoremstyle{remark}
\newtheorem{example}[theorem]{Example}
\newtheorem{remark}[theorem]{Remark}

\renewcommand{\leq}{\leqslant}
\renewcommand{\geq}{\geqslant}

\providecommand{\event}{TERMGRAPH 2014}
\newcommand{\titlerunning}{Presenting Finite Posets}
\newcommand{\authorrunning}{Samuel Mimram}
\title{\titlerunning}
\author{\authorrunning
\institute{CEA, LIST / École Polytechnique\thanks{This work was partially funded by the french ANR project CATHRE ANR-13-BS02-0005-02.}}
\email{\href{mailto:samuel.mimram@lix.polytechnique.fr}{samuel.mimram@lix.polytechnique.fr}}
}

\begin{document}
\maketitle

\begin{abstract}
  We introduce a monoidal category whose morphisms are finite partial orders,
  with chosen minimal and maximal elements as source and target
  respectively. After recalling the notion of presentation of a monoidal
  category by the means of generators and relations, we construct a presentation
  of our category, which corresponds to a variant of the notion of bialgebra.
\end{abstract}

String rewriting systems have been originally introduced by
Thue~\cite{thue1914probleme} in order to study word problems in monoids. A
string rewriting system $(\Sigma,R)$ consists of a set $\Sigma$, called the
\emph{alphabet}, and a set $R\subseteq\Sigma^*\times\Sigma^*$ of \emph{rules}.
The monoid $\Sigma^*/{\equiv_R}$, obtained by quotienting the free monoid
$\Sigma^*$ over $\Sigma$ by the smallest congruence (\wrt concatenation)
containing $R$, is called the monoid \emph{presented} by the rewriting system.
The rewriting system can thus be thought of as a small description of the
monoid, and the word problem consists in deciding whenever two words
$u,v\in\Sigma^*$ represent the same word, \ie are such that $u\equiv_Rv$. Now,
when the rewriting system is convergent, \ie both terminating and confluent,
normal forms provide canonical representatives of the equivalence classes: two
words $u,v\in\Sigma^*$ are equivalent by the congruence $\equiv_R$ if and only
if they have the same normal form, and the word problem can be thus be decided
in this case.

\begin{example}
  \label{ex:pres-monoid}
  Consider the rewriting system $(\Sigma,R)$ with $\Sigma=\set{a,b}$ and
  $R=\set{(ba,ab),(bb,\varepsilon)}$, where $\varepsilon$ denotes the empty
  word. This rewriting system is easily shown to be terminating, and the two
  critical pairs can be joined:
  \[
  \vxym{
    &\ar[dl]bba\ar[dr]&\\
    bab\ar@{.>}[r]&abb\ar@{.>}[r]&a
  }
  \qquad\qquad\qquad
  \vxym{
    bbb\ar@/_/[d]\ar@/^/[d]\\
    bb
  }
  \]
  the system is thus convergent. Normal forms are words of the form $a^n$ or
  $a^nb$, with $n\in\N$, and from this it is easy to deduce that the rewriting
  system presents the additive monoid $\N\times(\N/2\N)$: there is an obvious
  bijection between the set of normal forms and the elements of the monoid, and
  it remains to be checked that this bijection is compatible with product and
  units.
\end{example}

Starting from this point of view, it is natural to wonder whether the notion of
rewriting system can be extended to more general settings, in particular to
provide presentations for $n$-categories (a monoid can be seen as the particular
case of a $1$-category with only one object). A satisfactory answer to this
question was provided by Street and Power's
\emph{computads}~\cite{street1976limits,power1991n}, which were rediscovered as
Burroni's \emph{polygraphs}~\cite{burroni1993higher}, and really deserve the
name of \emph{higher-dimensional rewriting systems} at the light of the
preceding comments~\cite{mimram:trt}. Those were the subject of my introductory
invited presentation at the TERMGRAPH workshop in 2014, along with the way usual
rewriting techniques generalize to this framework.

The aim of this article is to describe a particular example of a presentation of
a monoidal category (\ie a 2-category with one 0-cell) by a 2-dimensional
rewriting system, whose morphisms are finite partial orders. This example was
discovered during the author's PhD thesis~\cite{mimram:phd} (Section~4.3.1)
while studying asynchronous game semantics. We will focus here on this example
only and introduce only the material which is necessary to handle it: we will
only define presentations for monoidal categories (see~\cite{burroni1993higher}
for the general definition of higher-dimensional rewriting system) and we will
not mention the links to game semantics
(see~\cite{mimram:phd,mellies2007asynchronous,mimram2011structure} for a more
detailed account of the relationships between game semantics, partial orders,
and presentations of categories). The proof has been split in a series of simple
enough lemmas, which is why we have left out most proofs.

\paragraph{Related works.}
The general methodology used here is strongly inspired from the one developed by
Lafont who studied many examples of 2-dimensional
presentations~\cite{lafont2003towards}. A variant of this example was studied by
Fiore and Devesas Campos~\cite{fiore2013algebra} using similar
techniques. Finally, a very interesting tool, based on distributive laws between
monads, has been introduced by Lack in order to build presentations of monoidal
categories by combining presentations of smaller (and simpler) monoidal
subcategories~\cite{lack2004composing}: a particularity of the present example
is that this technique does not apply here.


\paragraph{Plan.}
We begin by formally introducing the category~$\P$ of finite posets we will be
interested in giving a presentation in Section~\ref{sec:finpos}, detail what we
mean by a presentation of a monoidal category in Section~\ref{sec:presenting},
and give a few examples of presentations along with the one for $\P$ in
Section~\ref{sec:poalgebras}. The rest of the article is devoted to proving that
we actually have a presentation: we define a canonical factorization of
morphisms of $\P$ in Section~\ref{sec:cf}, which is used to finish the proof in
Section~\ref{sec:pres}. We conclude and hint at possible generalizations of this
result in Section~\ref{sec:concl}.

\section{A category of finite posets}
\label{sec:finpos}
A \emph{poset} is a pair $E=(\events{E},\leq_E)$ consisting of a set
$\events{E}$, whose elements are called \emph{events}, and a relation
${\leq_E}\subseteq\events{E}\times\events{E}$ which is reflexive, antisymmetric
and transitive. Two events $x,y\in\events{E}$ are \emph{independent} when
neither $x\leq_E y$ nor $y\leq_E x$.

\begin{definition}
  We write $\Pos$ for the category whose objects are posets and morphisms are
  increasing functions.
\end{definition}

\noindent
This category can be shown to be cocomplete~\cite{adamek1994locally}, and in
particular pushouts always exist. There is a full and faithful inclusion functor
$\Set\to\Pos$, such that the image of a set $X$ is the poset $(X,=)$, and we
will allow ourselves to implicitly see a set as a poset in this way: a poset in
the image of this functor is called \emph{discrete}. Given an integer~$n\in\N$,
we write $\intset{n}$ for the set $\set{0,1,\ldots,n-1}$ (or the associated
discrete poset), and $\dintset{n}$ for the totally ordered poset
$(\intset{n},\leq)$. Given a poset $E$ and a set $D\subseteq\events{E}$, we
write $E\setminus D$ for the poset such that $\events{E\setminus D}=\events
E\setminus D$ and $\leq_{E\setminus D}$ is the restriction of $\leq_E$ to this
set.

The main category of interest in this article is defined as follows:

\begin{definition}
  We write $\P$ for the category whose objects are natural numbers $n\in\N$, and
  morphisms $(s,E,t):m\to n$ are equivalence classes of triples
  consisting of
  \begin{itemize}
  \item a finite poset $E$,
  \item a monomorphism $s:[m]\to E$, called \emph{source}, whose images are
    minimal events in the poset,
  \item a monomorphism $t:[n]\to E$, called \emph{target}, whose images are
    maximal events in the poset,
  \end{itemize}
  quotiented by the equivalence relation identifying two such triples $(s,E,t)$
  and $(s',E',t')$ whenever there exists an isomorphism $f:E\to E'$ such that
  $f\circ s=s'$ and $f\circ t=t'$:
  \[
  \vxym{
    &E\ar[dd]_-f\\
    [m]\ar[ur]^s\ar[dr]_{s'}&&\ar[ul]_t\ar[dl]^{t'}[n]\\
    &E'
  }
  \]
  Given two composable morphisms $(s,E,t):m\to n$ and $(s',E',t'):n\to p$, we
  consider the poset $E+_{[n]}E'$ defined by the following pushout in $\Pos$
  \[
  \vxym{
    &E+_{[n]}E'\\
    E\ar@{.>}[ur]^-i&&\ar@{.>}[ul]_-{i'}E'\\
    &\ar[ul]^t[n]\ar[ur]_{s'}
  }
  \]
  and define their composite as
  \[
  (i\circ s,(E+_{[n]}E')\setminus[n],i'\circ t')
  \qcolon
  m\qto p
  \]
  The identity on an object $n$ is $(\id_n,[n],\id_n):n\to n$. We leave to the
  reader as an easy exercise to check that composition is well-defined and that
  the axioms of categories are satisfied (and more generally, we will leave to
  the reader the check that constructions subsequently performed on the category
  are compatible with the equivalence relation). An event $x\in\events E$ in a
  morphism $(s,E,t):m\to n$ is called \emph{external} if it is in the image of
  either $s$ or $t$, \ie $x\in s([m])\cup t([n])$, and \emph{internal}
  otherwise.
\end{definition}

\begin{remark}
  The above definition is a variant of the well-known construction of the
  bicategory of spans~\cite{benabou1967introduction} inside a category with
  pullbacks such as~$\Pos$.
\end{remark}

\begin{example}
  \label{ex:P-composition}
  The morphism $(s,E,t):2\to 3$ with $E=\set{a,b,c,d,e,f,g,h}$ with $a<c<f$,
  $a<d$ and $e<h$, $s(0)=a$, $s(1)=b$, $t(0)=f$, $t(1)=g$, $t(2)=h$ will be
  drawn by its Hasse diagram, \ie the graph of its immediate successor relation:
  \[
  \vxym{
    &\ar@{}[l]|-f\circ&\ar@{}[l]|-g\circ&\circ\ar@{}[r]|-h&\\
    &\ar@{}[l]|-c\bullet\ar@{-}[u]&\ar@{}[l]|-d\bullet&\bullet\ar@{-}[u]\ar@{}[r]|-e&\\
    &\ar@{}[l]|-a\circ\ar@{-}[u]\ar@{-}[u]\ar@{-}[ur]&&\circ\ar@{}[r]|-b&\\
  }
  \]
  The bullets represent events of the poset, the filled ones denoting internal
  events and empty ones denoting external events, and events are increasing from
  bottom to top. Notice that the names of the bullets do not really matter,
  because of the quotient used when defining morphisms, and we will not figure
  them in the following. The image of the source (\resp target) is figured by
  the empty bullets at the bottom (\resp top), where the $s(i)$ (\resp $t(i)$)
  are always represented with $i$ increasing from left to right. Composition is
  performed by ``gluing'' diagrams along the interface (by pushout) and erasing
  the external events of this interface:
  \[
  \pa{
    \vxym{
      \circ&\circ&\circ\\
      \bullet\ar@{-}[u]&\bullet&\bullet\ar@{-}[u]\\
      \circ\ar@{-}[u]\ar@{-}[u]\ar@{-}[ur]&&\circ\\
    }
  }
  \circ
  \pa{
    \vxym{
      \circ&&\circ\\
      \bullet\ar@{-}[u]&&\bullet\ar@{-}[u]\\
      \circ\ar@{-}[u]&\circ\ar@{-}[uul]&\circ\ar@{-}[u]\\
    }
  }
  \qeq
  \vxym{
    \circ&\circ&\circ\\
    \bullet\ar@{-}[u]&\bullet&\bullet\ar@{-}[u]\\
    \bullet\ar@{-}[u]\ar@{-}[ur]&&\bullet\\
    \circ\ar@{-}[u]&\circ\ar@{-}[uul]\ar@{-}[uu]&\circ\ar@{-}[u]\\
  }
  \]
\end{example}

Notice that this category is self-dual, in the sense that there is an
isomorphism $\P\cong\P^\op$, which is the identity on objects and respects the
monoidal structure introduced in next section
(Definition~\ref{def:P-monoidal}). Also, it contains the category of relations
as a subcategory:

\begin{definition}
  \label{def:rel}
  We write $\Rel$ for the category with integers as objects, such that a
  morphism $R:m\to n$ is a relation $R\subseteq\intset{m}\times\intset{n}$. The
  composite of two morphisms $R:m\to n$ and $S:n\to o$ is the relation $S\circ
  R:m\to o$ such that $(i,k)\in S\circ R$ if and only if there exists
  $j\in\intset{n}$ such that $(i,j)\in R$ and $(j,k)\in S$. The identity
  relation $\id_n:n\to n$ is such that $(i,j)\in\id_n$ if and only if $i=j$.
\end{definition}

\noindent
The following lemma will allow us to implicitly see a relation as a morphism in
the category $\P$:

\begin{lemma}
  \label{lemma:rel-poset}
  The following functor $\Rel\to\P$ is faithful. The functor is the identity on
  objects, and the image of a relation $R:m\to n$ is the morphism $(s,E,t):m\to
  n$, with $\events{E}=\intset{m}\uplus\intset{n}$, such that for any two
  elements $i\in\intset{m}$ and $j\in\intset{n}$ of $\events{E}$ we have $i\leq
  j$ if and only if $(i,j)\in R$, and the maps $s:\intset{m}\to E$ and
  $t:\intset{n}\to E$ are the injections of the coproduct.
\end{lemma}

\begin{example}
  \label{ex:rel-P}
  The relation $R:4\to 3$ with $R=\set{(0,0),(0,1),(0,2),(2,0)}$ can be seen as
  the following morphism of~$\P$:
  \[
  \vxym{
    \circ&\circ&\circ\\
    \circ\ar@{-}[u]\ar@{-}[ur]\ar@{-}[urr]&\circ&\ar@{-}[ull]\circ&\circ
  }
  \]
\end{example}

\section{Presenting monoidal categories}
\label{sec:presenting}
We recall that a \emph{strict monoidal category} $(\C,\otimes,I)$ consists of a
category $\C$ together with a functor $\otimes:\C\times\C\to\C$ called
\emph{tensor product} and an object $I\in\C$ called \emph{unit} such that the
tensor is associative and admits $I$ as unit: for every objects $A$, $B$ and
$C$, $(A\otimes B)\otimes C=A\otimes(B\otimes C)$ and $I\otimes A=A=A\otimes
I$. In this article, we only consider monoidal categories which are strict. A
\emph{symmetry} in such a category consists of a natural transformation of
components $\gamma_{A,B}:A\otimes B\to B\otimes A$ such that for every objects
$A$, $B$, $C$ we have
\[
\gamma_{A\otimes B,C}=(\gamma_{A,C}\otimes\id_B)\circ(\id_A\otimes\gamma_{B,C})
\qquad\qquad
\gamma_{I,A}=\id_A
\qquad\qquad
\gamma_{B,A}\circ\gamma_{A,B}=\id_{A\otimes B}
\]
A functor between monoidal categories is \emph{monoidal} when it respects the
tensor product and the unit, and we moreover suppose that monoidal functors
between symmetric monoidal categories preserve the symmetry. We write $\MonCat$
for the category of monoidal categories and monoidal functors.

\begin{definition}
  \label{def:P-monoidal}
  The category $\P$ can be made into a monoidal category with $0$ as unit, and
  tensor being defined by addition on objects ($m\otimes n=m+n$) and by disjoint
  union on morphisms. Moreover, a symmetry can easily be defined.
\end{definition}

\noindent
A monoidal category, such as $\P$, with integers as objects and tensor product
being given on objects by addition is often called a \emph{PRO}, or a
\emph{PROP} when it is additionally equipped with a
symmetry~\cite{maclane1965categorical}.

\begin{example}
  Using the diagrammatic notations of Example~\ref{ex:P-composition}, we have
  \[
  \pa{
    \vxym{
      \circ&\\
      \bullet\ar@{-}[u]&\\
      \circ\ar@{-}[u]&\circ\ar@{-}[uul]\\
    }
  }
  \otimes
  \pa{
    \vxym{
      \circ\\
      \bullet\ar@{-}[u]\\
      \circ\ar@{-}[u]\\
    }
  }
  \qeq
  \vxym{
    \circ&&\circ\\
    \bullet\ar@{-}[u]&&\bullet\ar@{-}[u]\\
    \circ\ar@{-}[u]&\circ\ar@{-}[uul]&\circ\ar@{-}[u]\\
  }
  \]
\end{example}

In order to define presentations of monoidal categories, we first need to define
an appropriate notion of signature for them. Since those contain objects and
morphisms, a signature will consist of generators for both of them. For the sake
of brevity, we follow here the formalization specific to monoidal categories
(also sometimes called \emph{tensor schemes}~\cite{joyal1991geometry}),
see~\cite{burroni1993higher} for the general case.

\begin{definition}
  A \emph{(monoidal) signature} $\Sigma=(\Sigma_1,s_1,t_1,\Sigma_2)$ consists of
  \begin{itemize}
  \item a set $\Sigma_1$ of \emph{object generators},
  \item a set $\Sigma_2$ of \emph{morphism generators},
  \item two functions $s_1,t_1:\Sigma_2\to\Sigma_1^*$, where $\Sigma_1^*$
    denotes the free monoid over $\Sigma_1$, assigning to a morphism generator
    its \emph{source} and \emph{target} respectively.
  \end{itemize}
  The category $\MonSig$ has signatures as objects and a morphism
  $f:(\Sigma_1,s_1,t_1,\Sigma_2)\to(\Sigma'_1,s'_1,t'_1,\Sigma'_2)$ is a pair
  $f=(f_1,f_2)$ of functions $f_1:\Sigma_1\to\Sigma'_1$ and
  $f_2:\Sigma_2\to\Sigma'_2$ such that $s_1'\circ f_2=f_1^*\circ s_1$ and
  $t_1'\circ f_2=f_1^*\circ t_1$, where $f_1^*:\Sigma_1^*\to{\Sigma'_1}^*$ is
  the extension of $f_1$ as a morphism of monoids.
\end{definition}

A monoidal signature generates a free monoidal category in the following
sense. Given a category~$\C$, we write $\Mor(\C)$ for the class of morphisms
of~$\C$.

\begin{proposition}
  Consider the forgetful functor $U:\MonCat\to\MonSig$ sending a monoidal
  category $(\C,\otimes,I)$ to the signature
  $\Sigma=(\Sigma_1,s_1,t_1,\Sigma_2)$, with $\Sigma_1$ being the set of objects
  of $\C$, $\Sigma_2$ consisting of triples $(A_1\ldots A_p,f,B_1\ldots
  B_q)\in\Sigma_1^*\times\Mor(\C)\times\Sigma_1^*$ such that
  $f\in\C(A_1\otimes\ldots\otimes A_n,B_1\otimes\ldots\otimes B_p)$, and
  functions $s_1$ and $t_1$ are given by first and third projection
  respectively. This functor admits a left adjoint $F:\MonSig\to\MonCat$. Given
  a signature $\Sigma$, the monoidal category $F\Sigma$ is often denoted
  $\Sigma^*$ and called the \emph{free monoidal category} on the signature.
  More explicitly, $\Sigma^*$ is the smallest category whose class of objects is
  $\Sigma_1^*$, for every morphism generator $f\in\Sigma_2$ there is a morphism
  $f\in\Sigma^*(s_1(f),t_1(f))$, for every object there is a formal identity on
  this object, morphisms are closed under formal composites and tensor products,
  and are quotiented by the axioms of monoidal categories.
\end{proposition}

\begin{remark}
  \label{rem:free-moncat}
  Suppose fixed a signature $\Sigma=(\Sigma_1,s_1,t_1,\Sigma_2)$ and a monoidal
  category $(\C,\otimes,I)$. Suppose moreover that we are given an object $FA$
  for every object generator $A\in\Sigma_1$, and a morphism
  $F\alpha:FA_1\otimes\ldots\otimes FA_p\to FB_1\otimes\ldots\otimes FB_q$ for
  every morphism generator $\alpha\in\Sigma_2$ with $s_1(\alpha)=A_1\ldots A_p$
  and $t_1(\alpha)=B_1\ldots B_q$. Then, by the freeness property of $\Sigma^*$,
  the operation $F$ extends uniquely as a monoidal functor $F:\Sigma^*\to\C$.
\end{remark}


\begin{remark}
  \label{rem:1-object}
  When the set $\Sigma_1=\set{1}$ is reduced to one element, $\Sigma_1^*$ is
  isomorphic to $\N$, and we denote by integers its elements. In this case, the
  monoidal category $\Sigma^*$ is a PRO.
\end{remark}

\noindent
The following lemma is often useful in order to prove properties of morphisms of
$\Sigma^*$ by induction over the number of generators they consist in:

\begin{lemma}
  \label{lemma:free-moncat-slice}
  Every morphism $f$ of $\Sigma^*$ can be written as a finite composite
  $f=f_k\circ\ldots\circ f_1$ where each morphism~$f_i$ is of the form
  $f_i=\id_{A_i}\otimes\alpha_i\otimes\id_{B_i}$ for some objects $A_i$ and
  $B_i$ and morphism generator~$\alpha_i\in\Sigma_2$.
\end{lemma}

The notion of rewriting system, adapted to monoidal categories, can finally be
defined as follows. We say that two morphisms in a category are \emph{parallel}
when they have the same source, and the same target.

\begin{definition}
  A \emph{(monoidal) rewriting system} is a pair $(\Sigma,R)$ consisting of a
  signature $\Sigma$ and a set $R\subseteq\Mor(\Sigma^*)\times\Mor(\Sigma^*)$ of
  pairs of morphisms which are parallel. The monoidal category \emph{presented}
  by such a rewriting system is the monoidal category $\Sigma^*/{\equiv_R}$
  obtained by quotienting the morphisms of~$\Sigma^*$ by the smallest
  congruence~$\equiv_R$, \wrt composition and tensor product, containing~$R$.
\end{definition}

\begin{remark}
  \label{rem:free-moncat-rel}
  Since $\equiv_R$ is the congruence generated by~$R$, a functor
  $F:\Sigma^*\to\C$ such that $F(f)=F(g)$ for every $(f,g)\in R$ induces a
  quotient functor $(\Sigma^*/{\equiv_R})\to\C$.
\end{remark}

\noindent
We often write $\alpha\To\beta$ for a rule $(\alpha,\beta)\in R$. We recall
below the most well-known example of presentation of a monoidal category: the
simplicial category. It is fundamental since it is at the heart of simplicial
algebraic topology.

\begin{example}
  \label{ex:pres-moncat}
  Consider the signature $\Sigma=(\Sigma_1,s_1,t_1,\Sigma_2)$ with
  $\Sigma_1=\set{1}$, $\Sigma_2=\set{\eta,\mu}$, $s_1(\eta)=\varepsilon$
  (where~$\varepsilon$ denotes the empty word), $s_1(\mu)=11$, and
  $t_1(\eta)=t_1(\mu)=1$. As explained in Remark~\ref{rem:1-object}, since the
  set~$\Sigma_1$ contains only one generator, the monoid $\Sigma_1^*$ is
  isomorphic to $\N$. The morphism generators can thus be denoted
  \[
  \eta:0\to 1
  \qquad\text{and}\qquad
  \mu:2\to 1
  \]
  along with their source and target. We will use this notation for defining
  morphism generators in the following. We consider the set of rules defined by
  \begin{equation}
    \label{eq:monoid-rules}
    R\qeq\set{
      \quad
      \mu\circ(\eta\otimes\id_1)\To\id_1
      \qcomma
      \mu\circ(\id_1\otimes\eta)\To\id_1
      \qcomma
      \mu\circ(\mu\otimes\id_1)\To\mu\circ(\id_1\otimes\mu)
      \quad
    }
  \end{equation}
  If we use the string diagrammatic notation for morphisms, and draw $\eta$ and
  $\mu$ respectively as on the left below, the rules can be figured as on the
  right:
  \[
  \strid{mult}
  \qquad
  \strid{unit}
  \qquad\qquad
  \strid{lab_unit_l}\To\strid{lab_unit_c}
  \qquad
  \strid{lab_unit_r}\To\strid{lab_unit_c}
  \qquad
  \strid{lab_assoc_l}\To\strid{lab_assoc_r}
  \]
  (notice that diagrams have to be read from bottom up, in order to be
  consistent with the notation for posets of Section~\ref{sec:finpos}). The
  rewriting system can be shown to be terminating and the five critical pairs
  can be joined, it is thus convergent. Normal forms are tensor products of
  ``right combs'', \eg
  \begin{equation}
    \label{eq:ex-rcomb}
    \strid{rcomb}
  \end{equation}
  Now, consider the \emph{simplicial category}~$\Delta$, with integers as
  objects, and such that a morphism $f:m\to n$ is an increasing function
  $f:[m]\to[n]$. It is monoidal when equipped with the expected tensor product
  given on objects by addition, \ie it is a PRO. With the above graphical
  representation in mind, it is easy to see that normal forms for the rewriting
  system are in bijection with morphisms of~$\Delta$, and that this bijection is
  compatible with composition and tensor product. For instance, the normal
  form~\eqref{eq:ex-rcomb} corresponds to the function $f:[9]\to[3]$ such that
  $f(i)=0$ if $0\leq i\leq 4$, and $f(i)=2$ if $5\leq i\leq 8$,
  see~\cite{lafont2003towards, mimram:trt} for details.
\end{example}

The rewriting system described in previous example corresponds to the theory of
monoids: from the above presentation, we immediately deduce that, given a
monoidal category~$\C$, monoidal functors $\Delta\to\C$ are in bijection with
monoids in~$\C$, see Definition~\ref{def:monoid} and~\cite{maclane:cwm} (and
moreover monoidal natural transformations correspond to morphisms of
monoids). We will thus call the rewriting system the \emph{theory} of monoids
(the ``the'' is slightly abusive since there are multiple choices for
orientations of rules and even of generators, but those will give the same
presented category).

The proof that the theory of monoids is a presentation of the simplicial
category, that we sketched in Example~\ref{ex:pres-moncat}, is a direct
adaptation of the classical rewriting argument in the case of monoids presented
in Example~\ref{ex:pres-monoid}. There are some theories for which no
orientation of the rules gives rise to a convergent rewriting system. However,
what matters here is only the fact that we have canonical representatives for
equivalence classes of morphisms in $\Sigma^*$ modulo the congruence generated
by the rules: we can hope to define them ``by hand'' (in which case we call them
\emph{canonical forms}), instead of obtaining them as normal forms for a
rewriting system. In order to show that a rewriting system~$(\Sigma,R)$ is a
presentation for a monoidal category~$\C$, one can thus apply the following
recipe discovered by Lafont~\cite{lafont2003towards}:
\begin{enumerate}
\item define a functor $\Sigma^*\to\C$ which is bijective on objects by
  interpreting generators of $\Sigma$ into $\C$, see
  Remark~\ref{rem:free-moncat},
\item show that it induces a functor $F:(\Sigma^*/{\equiv_R})\to\C$ by verifying
  that every pair of morphisms in $R$ have the same image, see
  Remark~\ref{rem:free-moncat-rel},
\item define (\eg inductively) a subset of morphisms in $\Sigma^*$, whose
  elements are called canonical forms,
\item show that every morphism in $\Sigma^*$ is equivalent by $\equiv_R$ to a
  canonical form by induction, see Lemma~\ref{lemma:free-moncat-slice},
\item show that the functor~$F$ restricted to canonical forms is a bijection.
\end{enumerate}
The first two steps allow us to define a functor $F:(\Sigma^*/{\equiv_R})\to\C$,
interpreting the presented monoidal category into~$\C$ and step 3 and 4 amount
to choose at least one representative in each equivalence class of morphisms
modulo $\equiv_R$. Finally, step 5 allows us to conclude. Namely, the
functor~$F$ is full since every morphism of $\C$ is the image of the equivalence
class of the corresponding canonical form. The functor is also faithful: given
two morphisms $f$ and $g$ such that $Ff=Fg$, by step~4 there is a canonical
form~$\ol f$ associated to $f$ and a canonical form $\ol g$ associated to~$g$
and those are necessarily equal by step~5, and therefore $f=g$. Notice that we
can conclude \emph{a posteriori} that each equivalence class contains exactly
one canonical form.
In the following, we use a variant of this methodology in order to build a
presentation for the monoidal category~$\P$. We will define the presentation and
interpret it in Section~\ref{sec:pres}.

\section{The theory of poalgebras}
\label{sec:poalgebras}
In this section, we first introduce some classical algebraic structures and
provide the monoidal category presented by the associated theory (\ie rewriting
system), and then define the algebraic structure corresponding to the
category~$\P$, which is a variant of those. We suppose fixed a monoidal
category~$(\C,\otimes,I)$ equipped with a symmetry~$\gamma$.

\begin{definition}
  \label{def:monoid}
  A \emph{monoid} $(M,\eta,\mu)$ in $\C$ consists of an object $M$ together with
  two morphisms $\eta:I\to M$ and $\mu:M\otimes M\to M$ such that
  $\mu\circ(\mu\otimes\id_M)=\mu\circ(\id_M\otimes\mu)$ and
  $\mu\circ(\eta\otimes\id_M)=\id_M=\mu\circ(\id_M\otimes\eta)$. Such a monoid
  is \emph{commutative} when $\mu\circ\gamma_{M,M}=\mu$. A \emph{comonoid}
  $(M,\varepsilon,\delta)$ is a monoid in $\C^\op$.
\end{definition}

\begin{proposition}[\cite{maclane:cwm, lafont2003towards}]
  \label{prop:free-monoid}
  The theory for monoids presents the simplicial category~$\Delta$, see
  Example~\ref{ex:pres-moncat}. The theory for commutative monoids presents the
  PROP~$\category{F}$ such that a morphism $f:m\to n$ is a function
  $f:[m]\to[n]$.
\end{proposition}

\begin{definition}
  \label{def:bialg}
  A \emph{bialgebra} $(B,\eta,\mu,\varepsilon,\delta)$ consists of a monoid
  $(B,\eta,\mu)$ and a comonoid $(B,\varepsilon,\delta)$ such that
  $\delta\circ\mu=(\mu\otimes\mu)\circ(\id_B\otimes\gamma_{B,B}\otimes\id_B)\circ(\delta\otimes\delta)$
  and $\varepsilon\circ\eta=\id_I$. Such a bialgebra is \emph{bicommutative}
  when both the monoid and the comonoid structure are commutative, and
  \emph{qualitative} when moreover $\mu\circ\delta=\id_B$.
\end{definition}

\begin{proposition}[\cite{maclane1965categorical, hyland2000symmetric, pirashvili2001prop, lafont2003towards, lack2004composing, mimram2011structure}]
  \label{prop:free-bialgebra}
  The theory for bicommutative bialgebras presents the PROP~$\category{Mat}_\N$
  such that a morphism $f:m\to n$ is an $(m\times n)$-matrix with coefficients
  in $\N$ together with usual composition. The theory for qualitative
  bicommutative bialgebras presents the PROP~$\Rel$, see
  Definition~\ref{def:rel}.
\end{proposition}

\begin{example}
  Consider the rewriting system $(\Sigma,R)$ corresponding to the theory for
  qualitative bicommutative bialgebras. The relation of Example~\ref{ex:rel-P}
  corresponds to the following morphism of $\Sigma^*/{\equiv_R}$:
  \[
  (\mu\otimes\id_2)
  \circ
  (\id_1\otimes\gamma\otimes\id_1)
  \circ
  (\delta\otimes\gamma)
  \circ
  (\delta\otimes\varepsilon\otimes\id_1\otimes\varepsilon)
  \]
  which can be represented graphically as
  $
  \strid{rel-bialg}
  $.
\end{example}

In the above propositions, we have been considering presentations of PROPs and
not PROs as earlier. In order to present them we could have used a variant of
rewriting systems in order to freely generate symmetries. However, thanks to the
following proposition~\cite{burroni1993higher}, there is a way to explicitly
incorporate an explicit symmetry into the presentations without changing the
notion of presentation. This is the one we have been implicitly using.

\begin{proposition}
  Suppose that a PRO $\C$ is presented by a rewriting system~$(\Sigma,R)$. Then
  the rewriting system $(\Sigma',R')$ presents the free PROP on the PRO~$\C$
  where $\Sigma'$ is obtained by adding a new morphism generator $\gamma:2\to 2$
  to $\Sigma$, and $R'$ is obtained from $R$ by adding the rule
  $\gamma\circ\gamma\To\id_2$ together with, for every morphism generator
  $\alpha:m\to n$, two rules
  \begin{equation}
    \label{eq:gamma-nat-rules}
    \gamma_{n,1}\circ(\alpha\otimes\id_1)\To(\id_1\otimes\alpha)\circ\gamma_{m,1}
    \qquad\qquad
    (\alpha\otimes\id_1)\circ\gamma_{1,m}\To\gamma_{1,n}\circ(\id_1\otimes\alpha)
  \end{equation}
  where the morphism $\gamma_{m,n}:m+n\to n+m$ is defined by induction on
  $(m,n)\in\N\times\N$ by
  \[
  \gamma_{1,0}=\id_1
  \qquad
  \gamma_{1,n+1}=(\id_1\otimes\gamma_{1,n})\circ(\gamma\otimes\id_n)
  \qquad
  \gamma_{0,n}=\id_n
  \qquad
  \gamma_{m+1,n}=(\gamma_{m,n}\otimes\id_n)\circ(\id_m\otimes\gamma_{1,n})
  \]
  Graphically, a representation of the rules~\eqref{eq:gamma-nat-rules} is
  \[
  \strid{gn1l}\To\strid{gn1r}
  \qquad\qquad
  \strid{gn2l}\To\strid{gn2r}
  \]
\end{proposition}

\begin{example}
  The PROP corresponding to monoids is presented by the rewriting system
  $(\Sigma,R)$ with $\Sigma_1=\set{1}$, and $\Sigma_2=\set{\eta:0\to 1,\mu:2\to
    1,\gamma:2\to 2}$, and $R$ consisting of rules~\eqref{eq:monoid-rules}
  together with
  \begin{align*}
    \gamma\circ(\eta\otimes\id_1)&\To\id_1\otimes\eta
    &
    \gamma\circ(\mu\otimes\id_1)&\To(\id_1\otimes\mu)\circ(\gamma\otimes\id_1)\circ(\id_1\otimes\gamma)
    \\
    \eta\otimes\id_1&\To\gamma\circ(\id_1\otimes\eta)
    &
    (\mu\otimes\id_1)\circ(\id_1\otimes\gamma)\circ(\gamma\otimes\id_1)&\To\gamma\circ(\id_1\otimes\mu)
    \\
    \gamma\circ\gamma&\To\id_2
    &
    (\gamma\otimes\id_1)\circ(\id_1\otimes\gamma)\circ(\gamma\otimes\id_1)&\To(\id_1\otimes\gamma)\circ(\gamma\otimes\id_1)\circ(\id_1\otimes\gamma)
  \end{align*}
  A presentation for the PROP~$\category{F}$ corresponding to commutative
  monoids is obtained by further adding the rule $\mu\circ\gamma\To\mu$, see
  Proposition~\ref{prop:free-monoid}.
\end{example}

The aim of this article is to show that the category~$\P$, defined in
Section~\ref{sec:finpos}, admits the following theory as presentation.

\begin{definition}
  A \emph{poalgebra} $(P,\eta,\mu,\varepsilon,\delta,\sigma)$ in a symmetric
  monoidal category $\C$ consists of a qualitative bicommutative bialgebra
  $(P,\eta,\mu,\varepsilon,\delta)$ in $\Sigma$ together with a morphism
  $\sigma:P\to P$ which is satisfying the following axiom, called
  \emph{transitivity}:
  \begin{equation}
    \label{eq:transitivity}
    \mu\circ(\id_P\otimes\sigma)\circ\delta\qeq\sigma
    \qquad\qquad\qquad\qquad
    \strid{lab_trans_l}\qeq\strid{lab_trans_r}
  \end{equation}
\end{definition}

\begin{remark}
  The additional axiom seems to break the horizontal symmetry of the structure,
  but this is not the case since the the dual axiom is implied:
  \[
  \mu\circ(\sigma\otimes\id_P)\circ\delta
  =\mu\circ\gamma\circ(\sigma\otimes\id_P)\circ\delta
  =\mu\circ(\id_P\otimes\sigma)\circ\gamma\circ\delta
  =\mu\circ(\id_P\otimes\sigma)\circ\delta
  =\sigma
  \]
\end{remark}

To sum up the axioms introduced in previous section, the theory of poalgebras is
the rewriting system $(\Sigma,R)$ with $\Sigma_1=\set{1}$ as object generators,
and the morphism generators in $\Sigma_2$ are, together with their diagrammatic
notation,
\[
\begin{array}{c@{\qquad}c@{\qquad}c@{\qquad}c@{\qquad}c@{\qquad}c}
  \eta:0\to 1
  &
  \mu:2\to 1
  &
  \varepsilon:1\to 0
  &
  \delta:1\to 2
  &
  \sigma:1\to 1
  &
  \gamma:2\to 2
  \\
  \strid{eta}
  &
  \strid{mu}
  &
  \strid{epsilon}
  &
  \strid{delta}
  &
  \strid{sigma}
  &
  \strid{gamma}
\end{array}
\]
and the rules in $R$ are
\begin{gather*}
  \begin{align*}
    \mu\circ(\eta\otimes\id_1)&\To\id_1
    &
    \mu\circ(\id_1\otimes\eta)&\To\id_1
    &
    \mu\circ(\mu\otimes\id_1)&\To\mu\circ(\id_1\otimes\mu)
    &
    \mu\circ\gamma&\To\mu
    \\
    (\varepsilon\otimes\id_1)\circ\delta&\To\id_1
    &
    (\id_1\otimes\varepsilon)\circ\delta&\To\id_1
    &
    (\delta\otimes\id_1)\circ\delta&\To(\id_1\otimes\delta)\circ\delta
    &
    \gamma\circ\delta&\To\delta
  \end{align*}
  \\
  \begin{align*}
  \delta\circ\mu&\To(\mu\otimes\mu)\circ(\id_1\otimes\gamma\otimes\id_1)\circ(\delta\otimes\delta)
  &
  \varepsilon\circ\eta&\To\id_0
  &
  \mu\circ\delta&\To\id_1
  &
  \mu\circ(\id_1\otimes\sigma)\circ\delta&\To\sigma
  \end{align*}
  \\
  \begin{align*}
    \gamma\circ\gamma&\To\id_2
    &
    (\gamma\otimes\id_1)\circ(\id_1\otimes\gamma)\circ(\gamma\otimes\id_1)&\To(\id_1\otimes\gamma)\circ(\gamma\otimes\id_1)\circ(\id_1\otimes\gamma)
  \end{align*}
  \\
  \begin{align*}
    \gamma\circ(\eta\otimes\id_1)&\To\id_1\otimes\eta
    &
    \eta\otimes\id_1&\To\gamma\circ(\id_1\otimes\eta)
    &
    \varepsilon\otimes\id_1&\To(\id_1\otimes\varepsilon)\circ\gamma
    &
    (\varepsilon\otimes\id_1)\circ\gamma&\To\id_1\otimes\varepsilon
  \end{align*}
  \\
  \begin{align*}
    \gamma\circ(\mu\otimes\id_1)&\To (\id_1\otimes\mu)\circ(\gamma\otimes\id_1)\circ(\id_1\otimes\gamma)
    &
    (\mu\otimes\id_1)\circ(\id_1\otimes\gamma)\circ(\gamma\otimes\id_1)&\To\gamma\circ(\id_1\otimes\mu)
    \\
    (\delta\otimes\id_1)\circ\gamma&\To(\id_1\otimes\gamma)\circ(\gamma\otimes\id_1)\circ(\id_1\otimes\delta)
    &
    (\gamma\otimes\id_1)\circ(\id_1\otimes\gamma)\circ(\delta\otimes\id_1)&\To(\id_1\otimes\delta)\circ\gamma
  \end{align*}
  \\
  \begin{align*}
    \gamma\circ(\sigma\otimes\id_1)&\To(\id_1\otimes\sigma)\circ\gamma
    &
    (\sigma\otimes\id_1)\circ\gamma&\To\gamma\circ(\id_1\otimes\sigma)
  \end{align*}
\end{gather*}
or graphically,
\renewcommand{\ss}[1]{\strid[scale=0.7]{#1}}
{\allowdisplaybreaks
\begin{align*}
  \ss{unit_l}&\To\ss{unit_c}&
  \ss{unit_r}&\To\ss{unit_c}&
  \ss{assoc_l}&\To\ss{assoc_r}&
  \ss{com_l}&\To\ss{com_r}
  \\
  \ss{counit_l}&\To\ss{counit_c}&
  \ss{counit_r}&\To\ss{counit_c}&
  \ss{coassoc_l}&\To\ss{coassoc_r}&
  \ss{cocom_l}&\To\ss{cocom_r}
  \\
  \ss{mu_delta_l}&\To\ss{mu_delta_r}&
  \ss{eta_eps_l}&\To\ss{eta_eps_r}&
  \ss{delta_mu_l}&\To\ss{delta_mu_r}&
  \ss{trans_l}&\To\ss{trans_r}
  \\
  &&
  \ss{inv_l}&\To\ss{inv_r}&
  \ss{yb_l}&\To\ss{yb_r}
  \\
  \ss{nat_etal_l}&\To\ss{nat_etal_r}&
  \ss{nat_etar_l}&\To\ss{nat_etar_r}&
  \ss{nat_epsr_l}&\To\ss{nat_epsr_r}&
  \ss{nat_epsl_l}&\To\ss{nat_epsl_r}&
  \\
  \ss{nat_mul_l}&\To\ss{nat_mul_r}&
  \ss{nat_mur_l}&\To\ss{nat_mur_r}&
  \ss{nat_deltal_l}&\To\ss{nat_deltal_r}&
  \ss{nat_deltar_l}&\To\ss{nat_deltar_r}
  \\
  &&
  \ss{nat_sigl_l}&\To\ss{nat_sigl_r}&
  \ss{nat_sigr_l}&\To\ss{nat_sigr_r}&
\end{align*}
}
This rewriting system is not confluent and whether the rules can be oriented and
completed in order to obtain a convergent rewriting system is still an open
question, although some progress have been made in this
direction~\cite{mimram:phd} (Section~5.5).
The main result of this article is the following theorem, and the rest of this
article is devoted to its proof.

\begin{theorem}
  The category~$\P$ is the PROP presented by the theory of poalgebras.
\end{theorem}

\section{Canonical factorizations}
\label{sec:cf}
In this section, we further study the category~$\P$, and define a canonical way
of writing a morphism in~$\P$ as a composite from a simple family of
morphisms. This will be used in Section~\ref{sec:pres} in order to show that
canonical forms of the presentation are in bijection with morphisms of $\P$. We
have omitted the proofs which can generally be done by easy (but not necessarily
short) induction.

The factorization of a morphism that we are going to introduce will not actually
depend on the morphism only, but on a linearization of it, in the following
sense: given a poset $E\in\Pos$, a \emph{linearization}~$\ell$ of $E$ is a
morphism $\ell:E\to\dintset{n}$ in $\Pos$, for some $n\in\N$, which is both epi
and mono.
While this definition is nice from a categorical point of view, it is sometimes
cumbersome to work with. It is easy to see that a morphism
$\ell:E\to\dintset{n}$ is a linearization of $E$ if and only if the underlying
function $\ell:\events{E}\to\intset{n}$ is a bijection (and thus $n$ is the
cardinal of $E$). A linearization of $E$ thus amounts to an enumeration
$\events{E}=\set{x_0,\ldots,x_{n-1}}$ of the events of $E$ (writing
$\ell^{-1}(i)$ as $x_i$), in way compatible with the partial order, and it turns
out to be simpler to axiomatize the inverse function. We will thus use the
following definition in the rest of the article:

\begin{definition}
  \label{def:linearization}
  A \emph{linearization} of a poset $E$ is a bijective function
  $x:\intset{n}\to\events{E}$ such that, for every $i,j\in\intset{n}$ with
  $i<j$, we have $x(i)\not\geq x(j)$.
\end{definition}

\noindent
The following lemma can be shown by induction on the cardinal of finite
posets. It is actually still true for non-finite posets if we assume the axiom
of choice~\cite{szpilrajn1930extension}, but we will not need this here.

\begin{lemma}
  Every finite poset admits a linearization.
\end{lemma}

\noindent
In order to characterize when two morphisms are linearizations of a given
poset~$E$, we define a relation~$\sim$ on linearizations of~$E$, as
follows. Given two linearizations $x,x':\intset{n}\to\events{E}$ of $E$, we have
$x\sim x'$ when there exists $i\in\intset{n-1}$ such that $x'(i)=x(i+1)$,
$x'(i+1)=x(i)$ and $x'(j)=x(j)$ for $j\neq i$ and $j\neq i+1$, \ie $x$ and $x'$
only differ by swapping two consecutive independent events; otherwise said,
writing $\tau^n_i:\intset{n}\to\intset{n}$ for the transposition exchanging $i$
and $i+1$ in the set $\intset{n}$, we have $x'=x\circ\tau^n_i$.

\begin{lemma}
  \label{lemma:lin-swap}
  Any two linearizations of a poset~$E$ are equivalent by the equivalence
  relation generated by~$\sim$.
\end{lemma}

\noindent
In the following, by a \emph{linearization of the internal events} of a morphism
$(s,E,t):m\to n$ in $\P$, we mean a linearization of the poset
$E\setminus\pa{s\pa{\intset{m}}\cup t\pa{\intset{n}}}$, or equivalently an
injective function $\intset{k}\to\events{E}$ whose image is the set of internal
events of~$E$ and which satisfies the condition of
Definition~\ref{def:linearization}.

We now introduce a notion of ``canonical factorization'' for morphisms of $\P$,
which will depend on a linearization of the internal events of the morphism.
Given $n\in\N$ and $I\subseteq[n]$, we write $X^n_I:n\to n+1$ for the morphism
with one internal node~$x$, so that
$\events{X}^n_I=[n]\uplus\set{x}\uplus[n+1]$, with the canonical injections as
source~$s:[n]\to X^n_I$ and target~$t:[n+1]\to X^n_I$, and whose partial order
is such that
\[
\forall i\in[n],\ s(i)<t(i)
\qqtand
\forall i\in I,\ s(i)<x
\qqtand
x<t(n)
\]

\begin{example}
  \label{ex:X302}
  For instance, the morphism $X^3_{\set{0,2}}:3\to 4$ is
  $
  \vxym{
    \circ&\circ&\circ&\circ\\
    &&&\ar@{-}[dlll]\ar@{-}[dl]\bullet\ar@{-}[u]\\
    \circ\ar@{-}[uu]&\circ\ar@{-}[uu]&\circ\ar@{-}[uu]\\
  }
  $.
\end{example}

\noindent
These morphisms can be used as basic ``building blocks'' for morphisms as follows.

\begin{lemma}
  \label{lemma:poset-cf-comp}
  Suppose given
  \begin{equation}
    \label{eq:fact}
    m,k,n\in\N,
    \qquad
    \text{sets $I_j\subseteq [m+j]$ indexed by $j\in[k]$},
    \qquad
    \text{and a relation $R:m+k\to n$.}
  \end{equation}
  The composite
  \begin{equation}
    \label{eq:fact-comp}
    R\circ X^{m+k-1}_{I_{k-1}}\circ\ldots\circ X^{m+1}_{I_1}\circ X^m_{I_0}
  \end{equation}
  is the morphism $(s,E,t):m\to n$ with $k$ internal events, that we denote by
  $x(i)$ with $i\in[k]$, which is the smallest partial order such that
  $x(i)<x(j)$ when $m+i\in I_j$, $s(i)<x(j)$ when $i\in I_j$, $x(i)<t(j)$ when
  $(m+i,j)\in R$, and $s(i)<t(j)$ when $(i,j)\in R$. Moreover, the function
  $x:\intset{k}\to\events{E}$ thus defined is a linearization of the internal
  events of~$E$.
\end{lemma}

\noindent
The data \eqref{eq:fact} is called a \emph{factorization} of the
morphism~\eqref{eq:fact-comp}, the morphism~\eqref{eq:fact-comp} is called the
\emph{composite} of the factorization \eqref{eq:fact}, and the above
linearization of the internal events is said to be \emph{induced} by the
factorization. From the preceding lemma, it is easy to deduce that any morphism
of~$P$ can be put into this form:

\begin{lemma}
  \label{lemma:poset-cf-fact}
  Suppose given a morphism $(s,E,t):m\to n$ of~$\P$ with $k$ internal events,
  together with a linearization $x:\intset{k}\to\events{E}$ of the internal
  events of~$E$. We have
  \begin{equation}
    \label{eq:poset-fact}
    E \qeq R\circ X^{m+k-1}_{I_{k-1}}\circ\ldots\circ X^{m+1}_{I_1}\circ X^m_{I_0}
  \end{equation}
  with
  \[
  I_j\qeq\setof{i\in[m]}{s(i)<x(j)}\cup\setof{m+i\in\intset{m+k}}{x(i)<x(j)}
  \]
  and $R:m+k\to n$ is the relation defined by
  \[
  R\qeq \setof{(i,j)\in[m]\times[n]}{s(i)<t(j)}\cup\setof{(m+i,j)}{i\in[k],j\in[n], x(i)<t(j)}
  \]
\end{lemma}

\begin{example}
  \label{ex:poset-cf}
  The poset $E:1\to 2$ of the left of
  \[
  \vxym{
    &\circ&&\circ\\
    \\
    &\ar@{}[l]|-d\bullet\ar@{-}[uu]\\
    \\
    &\ar@{}[l]|-b\bullet\ar@{-}[uu]\ar@{-}[uuuurr]&&\bullet\ar@{-}[uuuu]\ar@{}[r]|-c&\\
    &&\ar@{}[l]|-a\bullet\ar@{-}[ul]\ar@{-}[ur]\\
    &&\circ\ar@{-}[u]\\
  }
  \qeq
  \vxym{
    \circ\ar@{-}[dr(1)]\ar@{-}[drr(1)]\ar@{-}[drrrr(1)]&\circ\ar@{-}[dl(1)]\ar@{-}[dr(1)]\ar@{-}[drr(1)]\\
    &&&&\\
    &&&&\ar@{-}[dllll(1)]\ar@{-}[dlll(1)]\ar@{-}[dll(1)]\bullet\ar@{-}[u(1)]\ar@{}[r]|-d&\\
    &&&\ar@{-}[dlll(1)]\ar@{-}[dll(1)]\bullet\ar@{-}[uu(1)]\ar@{}[r]|-c&\\
    &&\ar@{-}[dll(1)]\ar@{-}[dl]\bullet\ar@{-}[uuu(1)]\ar@{}[r]|-b&\\
    &\ar@{-}[dl]\bullet\ar@{-}[uuuuu]\ar@{}[r]|-a&\\
    \circ\ar@{-}[uuuuuu]&
  }
  \qeq
  \vxym{
    \circ\ar@{-}[dr(1)]\ar@{-}[drrr(1)]\ar@{-}[drrrr(1)]&\circ\ar@{-}[dl(1)]\ar@{-}[dr(1)]\ar@{-}[drr(1)]\\
    &&&&\\
    &&&&\ar@{-}[dllll(1)]\ar@{-}[dlll(1)]\ar@{-}[dl]\bullet\ar@{-}[u(1)]\ar@{}[r]|-d&\\
    &&&\ar@{-}[dlll(1)]\ar@{-}[dll(1)]\bullet\ar@{-}[uu(1)]\ar@{}[r]|-b&\\
    &&\ar@{-}[dll(1)]\ar@{-}[dl]\bullet\ar@{-}[uuu(1)]\ar@{}[r]|-c&\\
    &\ar@{-}[dl]\bullet\ar@{-}[uuuuu]\ar@{}[r]|-a&\\
    \circ\ar@{-}[uuuuuu]
  }
  \]
  admits the factorization
  \[
  R\circ X^4_{\set{0,1,2}}\circ X^3_{\set{0,1}}\circ X^2_{\set{0,1}}\circ X^1_{\set{0}}
  \quad\text{with}\quad
  R=\set{(0,0),(1,0),(2,0),(4,0),(0,1),(1,1),(2,1),(3,1)}
  \]
  as shown in the middle. Another possible factorization of the same morphism is
  depicted on the right.
\end{example}

\noindent
Because partial orders are transitive, by using Lemma~\ref{lemma:poset-cf-comp}
the factorizations~\eqref{eq:poset-fact} of morphisms of~$\P$ provided by
Lemma~\ref{lemma:poset-cf-fact} can easily be shown to be transitive, in the
following sense:

\begin{definition}
  \label{def:cf-trans}
  A factorization of a morphism of $\P$ of the form \eqref{eq:poset-fact} is
  \emph{transitive} when
  \begin{itemize}
  \item given $i\in[m+k]$ and $i',i''\in[k]$ such that $i\in I_{i'}$ and
    $m+i'\in I_{i''}$, we have $i\in I_{i''}$,
  \item given $i\in[m+k]$, $i'\in[k]$ and $i''\in[n]$ such that $i\in I_{i'}$ and
    $(m+i',i'')\in R$ we have $(i,i'')\in R$.
  \end{itemize}
  A transitive factorization of a morphism is called a \emph{canonical
    factorization}.
\end{definition}

\begin{lemma}
  The factorizations~\eqref{eq:poset-fact} of a morphism of~$\P$ provided by
  Lemma~\ref{lemma:poset-cf-fact} are transitive. Moreover, the factorization
  associated to the composite of a transitive factorization, with the induced
  linearization of the internal events, is itself.
\end{lemma}

\noindent
To sum up, we have just shown that the data~\eqref{eq:fact} of a
factorization~\eqref{eq:poset-fact} is a way to faithfully encode a morphism
of~$\P$ together with a linearization of its internal events:

\begin{proposition}
  \label{prop:poset-lin-cf}
  Pairs $(E,x)$, constituted of a morphism~$E$ of~$\P$ together with a
  linearization $x$ of its internal events, are in bijection with factorizations
  of the form~\eqref{eq:poset-fact} which are transitive.
\end{proposition}

Because a morphism might admit multiple linearizations of its internal events, a
morphism generally admits multiple
factorizations~\eqref{eq:poset-fact}. However, we know by
Lemma~\ref{lemma:lin-swap} that any two of them are related by $\sim$, and the
following lemma characterizes how replacing a linearization by another one
related by $\sim$ affects the factorization.

\begin{lemma}
  \label{lemma:P-switch}
  Suppose that $x:[k]\to E$ is a linearization of $E:m\to n$ and $i\in[k]$ is
  such that $x'=x\circ\tau^k_i$ is another linearization of~$E$. Suppose that
  \[
  R\circ X^{m+k-1}_{I_{k-1}}\circ\ldots X^m_{I_0}
  \qqtand
  R'\circ X^{m+k-1}_{I'_{k-1}}\circ\ldots\circ X^m_{I'_0}
  \]
  are the factorizations associated by Proposition~\ref{prop:poset-lin-cf} to
  $(E,x)$ and $(E,x')$ respectively. Then we have $I_j'=I_j$ for $0\leq j<i+1$,
  $I_j'=\tau^{m+j}_{m+i}(I_j)$ for $i+1\leq j<k$, and
  $R'=R\circ\tau^{m+k}_{m+i}$.
\end{lemma}

\begin{example}
  The two factorizations of the morphism given in Example~\ref{ex:poset-cf} only
  differ by a swapping of~$b$ and~$c$. It can easily be checked that the
  relation in the second factorization can be deduced from the one of the first
  by precomposing by the transposition~$\tau^5_2$.
\end{example}

\section{A presentation of the monoidal category of finite posets}
\label{sec:pres}
As announced earlier, we use the general methodology explained at the end
of Section~\ref{sec:presenting} in order to show this. We first define an
interpretation $\intp{-}:\Sigma^*\to\P$ by defining the interpretation of the
generators, as per Remark~\ref{rem:free-moncat}. The object generator
$1\in\Sigma_1$ is interpreted as the object $1$ of $\P$, and the interpretations
of the morphism generators are the following morphisms of $\P$:
\[
\begin{array}{c@{\qquad\qquad}c@{\qquad\qquad}c@{\qquad\qquad}c@{\qquad\qquad}c@{\qquad\qquad}c}
  \vxym{
    \circ\\
    \\
    {\phantom{\circ}}\\
  }
  &
  \vxym{
    &\circ&\\
    \\
    \circ\ar@{-}[uur]&&\ar@{-}[uul]\circ
  }
  &
  \vxym{
    {\phantom{\circ}}\\
    \\
    \circ
  }
  &
  \vxym{
    \circ&&\circ\\
    \\
    &\ar@{-}[uul]\circ\ar@{-}[uur]
  }
  &
  \vxym{
    \circ\\
    \bullet\ar@{-}[u]\\
    \circ\ar@{-}[u]
  }
  &
  \vxym{
    \circ&\circ\\
    \\
    \circ\ar@{-}[uur]&\circ\ar@{-}[uul]
  }
  \\[5ex]
  \intp{\eta}&\intp{\mu}&\intp{\varepsilon}&\intp{\delta}&\intp{\sigma}&\intp{\gamma}
\end{array}
\]
Moreover, the interpretation leaves the rules invariant:

\begin{lemma}
  For any rule $\alpha\To\beta$ of $R$, we have
  $\intp{\alpha}=\intp{\beta}$. Otherwise said,
  $(1,\intp{\eta},\intp{\mu},\intp{\varepsilon},\intp{\delta},\intp{\sigma})$ is
  a poalgebra in $\P$.
\end{lemma}

\noindent
We write $\TP$ for the category $\TP=\pa{\Sigma^*/{\equiv_R}}$. As noticed in
Remark~\ref{rem:free-moncat-rel}, by the above lemma the interpretation functor
$\intp{-}:\Sigma^*\to\P$ induces a quotient functor $\intp{-}:\TP\to\P$, which
is bijective on objects. Our aim in this section is to show that it is an
isomorphism, \ie $\TP\cong\P$.

Given $n\in\N$ and $i\in[n]$, we define notations for the following morphisms
in~$\TP$:
\begin{itemize}
\item $I^n=\id_{n+1}:n+1\to n+1$
\item $H^n=(\id_n\otimes\eta):n\to n+1$
\item $S^n=(\id_n\otimes\sigma):n+1\to n+1$
\item $G^n=\id_n\otimes\gamma_{1,1}:n+2\to n+2$
\item $W^n_i=(\id_n\otimes\mu)\circ(\id_{i+1}\otimes\gamma_{1,n-i-1}\otimes\id_1)\circ(\id_i\otimes\delta\otimes\id_{n-i}):n+1\to n+1$
\end{itemize}
Graphically,
\[
I^n=\strid[height=1cm]{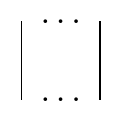}
\qquad
H^n=\strid[height=1cm]{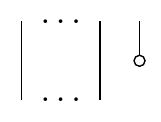}
\qquad
S^n=\strid[height=1cm]{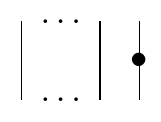}
\qquad
G^n=\strid[height=2cm]{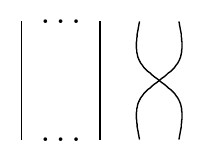}
\qquad
W^n_i=\strid[height=2cm]{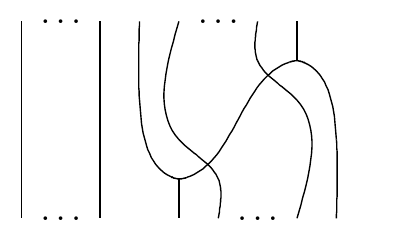}
\]

Using the relations in $R$ it is easy to show that
\begin{lemma}
  For every $n\in\N$, and $i,j\in[n]$, we have $W^n_j\circ W^n_i=W^n_i\circ W^n_j$ and
  $W^n_i\circ W^n_i=W^n_i$.
\end{lemma}

\noindent
The above lemma implies that given $I=\set{i_0,\ldots,i_{k-1}}\subseteq\intset{n}$,
it makes sense to define the morphism
\[
W^n_I\qeq W^n_{i_{k-1}}\circ\ldots\circ W^n_{i_1}\circ W^n_{i_0}\qcolon n+1 \qto n+1
\]
with, by convention, $W^n_\emptyset=I^n$. We also define for $n\in\N$ and
$I\subseteq[n]$, $X^n_I=S^n\circ W^n_I\circ H^n:n\to n+1$. The following lemma
justifies that we use the same notation for this morphism of $\TP$ and the
morphism of $\P$ introduced in Section~\ref{sec:cf}.

\begin{lemma}
  Given $n\in\N$ and $I\subseteq\intset{n}$, we have $\intp{X^n_I}=X^n_I$.
\end{lemma}

\begin{example}
  For instance $X^3_{0,2}$ is as follows (notice the resemblance with
  Example~\ref{ex:X302}):
  \[
  X^3_{0,2}
  \qeq
  \strid[height=2cm]{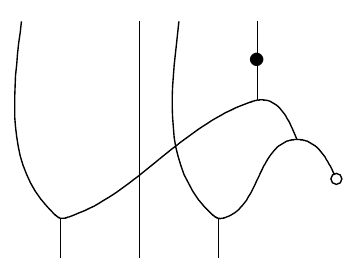}
  \]
\end{example}

\noindent
We say that a morphism $R$ is a \emph{relation} when it can be obtained by
composing and tensoring generators not involving~$\sigma$. The name is justified
by Proposition~\ref{prop:free-bialgebra} and the following lemma.

\begin{lemma}
  The canonical embedding $\Rel\to\TP$ is faithful.
\end{lemma}
\begin{proof}
  By definition of~$\TP$, this category contains a qualitative bicommutative
  bialgebra and therefore there is a canonical embedding functor $\Rel\to\TP$ by
  Proposition~\ref{prop:free-bialgebra}, whose image consists of morphisms which
  can be obtained by composing and tensoring generators not involving
  $\sigma$. Moreover, the property of ``not containing a generator~$\sigma$'' in
  a morphism of~$\TP$ is preserved by the axioms of poalgebras and the axioms of
  poalgebras not involving~$\sigma$ are precisely the axioms of qualitative
  bicommutative bialgebras.
\end{proof}

\begin{lemma}
  \label{lemma:tp-p-full}
  The canonical functor $\intp{-}:\TP\to\P$ is full.
\end{lemma}
\begin{proof}
  Suppose given a morphism $(s,E,t):m\to n$. By
  Lemma~\ref{lemma:poset-cf-fact}, the morphism admits a factorization of the
  form~\eqref{eq:poset-fact}. Since for any relation $R$ in $\TP$, $\intp{R}$ is
  the corresponding relation in $\P$, and we have $\intp{X^n_I}=X^n_I$, the
  morphism $E$ is in the image of $\intp{-}$.
\end{proof}

\begin{lemma}
  \label{lemma:tp-cf}
  Every morphism $f:m\to n$ of $\TP$ can be factored as
  \begin{equation}
    \label{eq:tp-nf}
    R\circ X^{m+k-1}_{I_{k-1}}\circ\ldots\circ X^{m+1}_{I_1}\circ X^m_{I_0}
  \end{equation}
  for some $k\in\N$ with $I_j\subseteq[n+j]$ for every $j\in[k]$ and $R:m+k\to
  n$ is a relation.
\end{lemma}
\begin{proof}
  We proceed by induction on the number of generators occurring in~$f$. If $f$
  is the identity, it has no internal event, \ie it is a relation as per
  Lemma~\ref{lemma:rel-poset}; it is thus of the form~\eqref{eq:tp-nf} with
  $k=0$. Otherwise, by Lemma~\ref{lemma:free-moncat-slice}, the morphism $f$ can
  be factored as $f=f'\circ(\id_i\otimes\alpha\otimes\id_j)$ with $i,j\in\N$ and
  $\alpha$ a morphism generator. By induction hypothesis, we know that $f'$
  admits a factorization of the form~\eqref{eq:tp-nf}, and we show that $f$ then
  admits a factorization~\eqref{eq:tp-nf} by case analysis on~$i$ and~$\alpha$.
\end{proof}

\noindent
We say that a morphism of the form~\eqref{eq:tp-nf} is \emph{transitive} when it
satisfies the conditions for transitivity of Definition~\ref{def:cf-trans}.

\begin{lemma}
  \label{lemma:TP-trans-fact}
  Every morphism $f:m\to n$ of $\TP$ admits a factorization of the
  form~\eqref{eq:tp-nf} which is transitive. Such a factorization is called
  \emph{canonical}.
\end{lemma}
\begin{proof}
  Suppose given a morphism $f:m\to n$. By Lemma~\ref{lemma:tp-cf}, we know that
  it admits a factorization of the form~\eqref{eq:tp-nf}. This factorization can
  be shown to be equal to a transitive one by progressively adding
  ``transitivity edges'' (in the sense of Definition~\ref{def:cf-trans}), by
  using the fact that the $\TP$ contains a poalgebra, which in particular
  satisfies the transitivity axiom~\eqref{eq:transitivity}.
\end{proof}


\noindent
Similarly to Section~\ref{sec:cf} (see Lemma~\ref{lemma:P-switch}), the
canonical factorizations are not unique but one can go from one to another by
``switching'' two consecutive $X^n_I$ in the factorization, which will result in
precomposing the relation of the factorization by a transposition. The following
lemmas formalize this.

\begin{lemma}
  Given $n\in\N$, and $I,J\subseteq [n]$, we have $X^{n+1}_J\circ X^n_I=G^n\circ
  X^{n+1}_I\circ X^n_J$.
\end{lemma}
\begin{proof}
  With the notations $\ol H^n=H^n\otimes\id_1$, $\ol S^n=S^n\otimes\id_1$, $\ol
  W^n_I=W^n_I\otimes\id_1$, we have
  \begin{align*}
    X^{n+1}_JX^n_I
    &\qeq S^{n+1}W^{n+1}_JH^{n+1}S^nW^n_IH^n&&\qeq \ol S^n\ol W^n_IS^{n+1}W^{n+1}_J\ol H^nH^n\\
    &\qeq S^{n+1}W^{n+1}_J\ol S^nH^{n+1}W^n_IH^n&&\qeq \ol S^n\ol W^n_IS^{n+1}\ol H^nW^n_JH^n\\
    &\qeq S^{n+1}W^{n+1}_J\ol S^n\ol W^n_IH^{n+1}H^n&&\qeq \ol S^n\ol W^n_IS^{n+1}\ol H^nW^n_JH^n\\
    &\qeq S^{n+1}W^{n+1}_J\ol S^n\ol W^n_I\ol H^nH^n&&\qeq \ol S^n\ol W^n_I\ol H^nS^nW^n_JH^n\\
    &\qeq S^{n+1}\ol S^nW^{n+1}_J\ol W^n_I\ol H^nH^n&&\qeq \ol S^n\ol W^n_IG^nH^{n+1}S^nW^n_JH^n\\
    &\qeq \ol S^nS^{n+1}W^{n+1}_J\ol W^n_I\ol H^nH^n&&\qeq \ol S^nG^nW^{n+1}_IH^{n+1}S^nW^n_JH^n\\
    &\qeq \ol S^nS^{n+1}\ol W^n_IW^{n+1}_J\ol H^nH^n&&\qeq G^nX^{n+1}_IX^n_J
  \end{align*}
  Each of the equalities involved are direct to show by using the definition of
  the morphisms and the axioms symmetric monoidal categories.
  For instance, one can show that $H^{n+1}S^n=\ol S^nH^{n+1}$ using the exchange
  law of monoidal categories: $
  H^{n+1}S^n=(n+1\otimes\eta)\circ(n\otimes\sigma)=(n\otimes\sigma\otimes
  1)\circ(n\otimes\eta)=\ol S^nH^{n+1} $.  Other steps are similar.
\end{proof}

\begin{lemma}
  We write $G^n_i=i\otimes\gamma\otimes(n-i)$ (so that $G^n=G^n_n$). Given
  $n\in\N$, $i\in[n+1]$ and $I\subseteq[n+1]$, we have
  $X^{n+2}_I\circ G^n_i=G^n_i\circ X^{n+1}_{\tau^{n+1}_i(I)}$.
\end{lemma}

\noindent
As a direct corollary of the two previous lemmas, we have

\begin{lemma}
  \label{lemma:TP-switch}
  Given $k\in\N$, sets $I_j\subseteq[n+j]$ indexed by $j\in[k]$, a relation
  $R:m+k\to n$, and $i\in[k-1]$ such that $m+i\not\in I_{i+1}$, we have $R\circ
  X^{m+k-1}_{I_{k-1}}\circ\ldots\circ X^m_{I_0} = R'\circ
  X^{m+k-1}_{I'_{k-1}}\circ\ldots\circ X^m_{I'_0}$ with $I'_j=I_j$ for $0\leq
  j<i+1$, $I'_j=\tau^{m+j}_{m+i}(I_j)$ for $i+1\leq j<k$, and
  $R'=R\circ\tau^{m+k}_{m+i}$.
\end{lemma}

\noindent
Finally, we are now able to show that the axioms of poalgebras are a
presentation of the category~$\P$.

\begin{theorem}
  \label{thm:p-qtbbe}
  The categories~$\TP$ and~$\P$ are isomorphic as symmetric monoidal categories:
  the category $\P$ is presented by the theory of poalgebras.
\end{theorem}
\begin{proof}
  The canonical functor $\intp{-}:\TP\to\P$ is bijective on objects and was
  shown to be full in Lemma~\ref{lemma:tp-p-full}. In order to conclude, we have
  to show that it is faithful. Suppose given two morphisms $f,g:m\to n$ of
  $\TP$, such that $\intp f=\intp g$. By Lemma~\ref{lemma:TP-trans-fact}, we can
  suppose that they are written as a canonical factorization, and their images
  by $\intp{-}$ are thus canonical factorizations in~$\P$, in the sense of
  Definition~\ref{def:cf-trans}. They therefore correspond to two linearizations
  of the same poset. By Lemma~\ref{lemma:lin-swap}, we can transform one
  permutation into the other by a series of commutations of consecutive
  elements, each one corresponding to transforming the canonical form in~$\P$ in
  the way described in Lemma~\ref{lemma:P-switch}. We write $\intp
  f=E_1=E_2=\ldots=E_n=\intp g$ for this sequence of canonical forms. By
  Lemma~\ref{lemma:TP-switch}, we can mimic these transformations on $f$,
  obtaining a sequence of canonical factorizations $f=f_1=f_2=\ldots=f_n=g$
  in~$\TP$ such that for every $i$ we have~$\intp{f_i}=E_i$. Therefore $f=g$ and
  the functor~$\intp-$ is faithful.
\end{proof}

\section{Conclusion}
\label{sec:concl}
We have shown that the monoidal category $\P$ of partial orders, is the theory
of poalgebras. This monoidal category contains as subcategories various
well-known categories such as the simplicial category $\Delta$, the category
$\category{F}$ of functions, the category $\Rel$ of relations, etc. and the
presentations of those categories (see Section~\ref{sec:poalgebras}) can thus be
recovered as particular cases of our result.

Many interesting variants of this result can be obtained along the same
lines. For instance, by adding the axioms $\varepsilon\circ\sigma=\varepsilon$
and $\sigma\circ\eta=\eta$ to the theory of poalgebra, we obtain a presentation
of the subcategory of $\P$ whose morphisms are of the form $(s,E,t)$ with $s$
and $t$ surjective. Similarly, by removing the transitivity
axiom~\eqref{eq:transitivity} we can present a variant of the category where
morphisms are partial orders which are not supposed to be transitive, which can
be seen as directed acyclic graph, and thus recover the result of Fiore and
Devesas Campos~\cite{fiore2013algebra}. A variant of this theory, where the
morphism~$\sigma$ has been replaced by a family of generators with various
number of inputs and outputs, has also been considered in~\cite{beauxis2011non}
in order to provide an axiomatic definition of the notion of ``network''.

This work should be considered as a first step and can be generalized in various
promising ways. This presentation can be extended from posets to \emph{event
  structures}~\cite{winskel1995models}, which are partial orders equipped with a
notion of incompatibility, thus allowing us to hope for algebraic methods in
order to study concurrent processes, which those structures naturally
model. Another very interesting direction consists in exploring
higher-dimensional structures: for instance, we would like to extend this work
to construct a presentation of the monoidal bicategory of integers, partial
orders and increasing functions. In particular the full monoidal subbicategory
whose 1-cells are forests (\ie partial orders such that two incomparable
elements do not have an upper bound) is expected to be presented by (a variant
of) the theory of commutative strong monads, and to be closely related to
Moerdijk's dendroidal sets~\cite{moerdijk2010dendroidal}.

\bibliographystyle{eptcs}
\bibliography{papers}

\begin{thebibliography}{10}
\providecommand{\bibitemdeclare}[2]{}
\providecommand{\surnamestart}{}
\providecommand{\surnameend}{}
\providecommand{\urlprefix}{Available at }
\providecommand{\url}[1]{\texttt{#1}}
\providecommand{\href}[2]{\texttt{#2}}
\providecommand{\urlalt}[2]{\href{#1}{#2}}
\providecommand{\doi}[1]{doi:\urlalt{http://dx.doi.org/#1}{#1}}
\providecommand{\bibinfo}[2]{#2}

\bibitemdeclare{book}{adamek1994locally}
\bibitem{adamek1994locally}
\bibinfo{author}{Ji{\v{r}}{\'\i} \surnamestart Ad{\'a}mek\surnameend} \&
  \bibinfo{author}{Ji{\v{r}}{\'\i} \surnamestart Rosicky\surnameend}
  (\bibinfo{year}{1994}): \emph{\bibinfo{title}{Locally presentable and
  accessible categories}}.
\newblock \bibinfo{volume}{189}, \bibinfo{publisher}{Cambridge University
  Press}, \doi{10.1017/CBO9780511600579}.

\bibitemdeclare{inproceedings}{beauxis2011non}
\bibitem{beauxis2011non}
\bibinfo{author}{Romain \surnamestart Beauxis\surnameend},
  \bibinfo{author}{Samuel \surnamestart Mimram\surnameend} et~al.
  (\bibinfo{year}{2011}): \emph{\bibinfo{title}{A Non-Standard Semantics for
  Kahn Networks in Continuous Time}}.
\newblock In: {\sl \bibinfo{booktitle}{Computer Science Logic (CSL'11)-25th
  International Workshop/20th Annual Conference of the EACSL}},
  \bibinfo{volume}{12}, pp. \bibinfo{pages}{35--50},
  \doi{10.4230/LIPIcs.CSL.2011.35}.

\bibitemdeclare{inproceedings}{benabou1967introduction}
\bibitem{benabou1967introduction}
\bibinfo{author}{Jean \surnamestart B{\'e}nabou\surnameend}
  (\bibinfo{year}{1967}): \emph{\bibinfo{title}{Introduction to bicategories}}.
\newblock In: {\sl \bibinfo{booktitle}{Reports of the Midwest Category
  Seminar}}, \bibinfo{organization}{Springer}, pp. \bibinfo{pages}{1--77},
  \doi{10.1007/BFb0074299}.

\bibitemdeclare{article}{burroni1993higher}
\bibitem{burroni1993higher}
\bibinfo{author}{Albert \surnamestart Burroni\surnameend}
  (\bibinfo{year}{1993}): \emph{\bibinfo{title}{Higher-dimensional word
  problems with applications to equational logic}}.
\newblock {\sl \bibinfo{journal}{Theoretical computer science}}
  \bibinfo{volume}{115}(\bibinfo{number}{1}), pp. \bibinfo{pages}{43--62},
  \doi{10.1016/0304-3975(93)90054-W}.

\bibitemdeclare{incollection}{fiore2013algebra}
\bibitem{fiore2013algebra}
\bibinfo{author}{Marcelo \surnamestart Fiore\surnameend} \&
  \bibinfo{author}{Marco~Devesas \surnamestart Campos\surnameend}
  (\bibinfo{year}{2013}): \emph{\bibinfo{title}{The algebra of directed acyclic
  graphs}}.
\newblock In: {\sl \bibinfo{booktitle}{Computation, Logic, Games, and Quantum
  Foundations. The Many Facets of Samson Abramsky}},
  \bibinfo{publisher}{Springer}, pp. \bibinfo{pages}{37--51},
  \doi{10.1007/978-3-642-38164-5\_4}.

\bibitemdeclare{inproceedings}{hyland2000symmetric}
\bibitem{hyland2000symmetric}
\bibinfo{author}{Martin \surnamestart Hyland\surnameend} \&
  \bibinfo{author}{John \surnamestart Power\surnameend} (\bibinfo{year}{2000}):
  \emph{\bibinfo{title}{Symmetric monoidal sketches}}.
\newblock In: {\sl \bibinfo{booktitle}{Proceedings of the 2nd ACM SIGPLAN
  international conference on Principles and practice of declarative
  programming}}, \bibinfo{organization}{ACM}, pp. \bibinfo{pages}{280--288},
  \doi{10.1145/351268.351299}.

\bibitemdeclare{article}{joyal1991geometry}
\bibitem{joyal1991geometry}
\bibinfo{author}{Andr{\'e} \surnamestart Joyal\surnameend} \&
  \bibinfo{author}{Ross \surnamestart Street\surnameend}
  (\bibinfo{year}{1991}): \emph{\bibinfo{title}{The geometry of tensor
  calculus, I}}.
\newblock {\sl \bibinfo{journal}{Advances in Mathematics}}
  \bibinfo{volume}{88}(\bibinfo{number}{1}), pp. \bibinfo{pages}{55--112},
  \doi{10.1016/0001-8708(91)90003-P}.

\bibitemdeclare{article}{lack2004composing}
\bibitem{lack2004composing}
\bibinfo{author}{Stephen \surnamestart Lack\surnameend} (\bibinfo{year}{2004}):
  \emph{\bibinfo{title}{Composing PROPs}}.
\newblock {\sl \bibinfo{journal}{Theory and Applications of Categories}}
  \bibinfo{volume}{13}(\bibinfo{number}{9}), pp. \bibinfo{pages}{147--163}.

\bibitemdeclare{article}{lafont2003towards}
\bibitem{lafont2003towards}
\bibinfo{author}{Yves \surnamestart Lafont\surnameend} (\bibinfo{year}{2003}):
  \emph{\bibinfo{title}{Towards an algebraic theory of boolean circuits}}.
\newblock {\sl \bibinfo{journal}{Journal of Pure and Applied Algebra}}
  \bibinfo{volume}{184}(\bibinfo{number}{2}), pp. \bibinfo{pages}{257--310},
  \doi{10.1016/S0022-4049(03)00069-0}.

\bibitemdeclare{article}{maclane1965categorical}
\bibitem{maclane1965categorical}
\bibinfo{author}{Saunders \surnamestart MacLane\surnameend}
  (\bibinfo{year}{1965}): \emph{\bibinfo{title}{Categorical algebra}}.
\newblock {\sl \bibinfo{journal}{Bulletin of the American Mathematical
  Society}} \bibinfo{volume}{71}(\bibinfo{number}{1}), pp.
  \bibinfo{pages}{40--106}, \doi{10.1090/S0002-9904-1965-11234-4}.

\bibitemdeclare{book}{maclane:cwm}
\bibitem{maclane:cwm}
\bibinfo{author}{Saunders \surnamestart MacLane\surnameend}
  (\bibinfo{year}{1998}): \emph{\bibinfo{title}{Categories for the Working
  Mathematician}}.
\newblock \bibinfo{series}{Graduate Texts in Mathematics},
  \bibinfo{publisher}{Springer}, \doi{10.1007/978-1-4757-4721-8}.

\bibitemdeclare{incollection}{mellies2007asynchronous}
\bibitem{mellies2007asynchronous}
\bibinfo{author}{Paul-Andr{\'e} \surnamestart Melli{\`e}s\surnameend} \&
  \bibinfo{author}{Samuel \surnamestart Mimram\surnameend}
  (\bibinfo{year}{2007}): \emph{\bibinfo{title}{Asynchronous games: innocence
  without alternation}}.
\newblock In: {\sl \bibinfo{booktitle}{CONCUR 2007--Concurrency Theory}},
  \bibinfo{publisher}{Springer}, pp. \bibinfo{pages}{395--411},
  \doi{10.1007/978-3-540-74407-8\_27}.

\bibitemdeclare{phdthesis}{mimram:phd}
\bibitem{mimram:phd}
\bibinfo{author}{Samuel \surnamestart Mimram\surnameend}
  (\bibinfo{year}{2008}): \emph{\bibinfo{title}{Sémantique des jeux
  asynchrones et réécriture 2-dimensionnelle}}.
\newblock Ph.D. thesis, \bibinfo{school}{Université Paris Diderot -- Paris
  VII}.

\bibitemdeclare{article}{mimram2011structure}
\bibitem{mimram2011structure}
\bibinfo{author}{Samuel \surnamestart Mimram\surnameend}
  (\bibinfo{year}{2011}): \emph{\bibinfo{title}{The structure of first-order
  causality}}.
\newblock {\sl \bibinfo{journal}{Mathematical Structures in Computer Science}}
  \bibinfo{volume}{21}(\bibinfo{number}{01}), pp. \bibinfo{pages}{65--110},
  \doi{10.1017/S0960129510000459}.

\bibitemdeclare{article}{mimram:trt}
\bibitem{mimram:trt}
\bibinfo{author}{Samuel \surnamestart Mimram\surnameend}
  (\bibinfo{year}{2014}): \emph{\bibinfo{title}{{Towards 3-Dimensional
  Rewriting Theory}}}.
\newblock {\sl \bibinfo{journal}{Logical Methods in Computer Science}}
  \bibinfo{volume}{10}(\bibinfo{number}{2}), pp. \bibinfo{pages}{1--47},
  \doi{10.2168/LMCS-10(2:1)2014}.

\bibitemdeclare{incollection}{moerdijk2010dendroidal}
\bibitem{moerdijk2010dendroidal}
\bibinfo{author}{Ieke \surnamestart Moerdijk\surnameend} \&
  \bibinfo{author}{Bertrand \surnamestart To{\"e}n\surnameend}
  (\bibinfo{year}{2010}): \emph{\bibinfo{title}{Dendroidal sets}}.
\newblock In: {\sl \bibinfo{booktitle}{Simplicial Methods for Operads and
  Algebraic Geometry}}, \bibinfo{publisher}{Springer}, pp.
  \bibinfo{pages}{23--39}, \doi{10.1007/978-3-0348-0052-5\_4}.

\bibitemdeclare{article}{pirashvili2001prop}
\bibitem{pirashvili2001prop}
\bibinfo{author}{Teimuraz \surnamestart Pirashvili\surnameend}
  (\bibinfo{year}{2002}): \emph{\bibinfo{title}{On the $ PROP $ corresponding
  to bialgebras}}.
\newblock {\sl \bibinfo{journal}{Cahiers de Topologie et G{\'e}om{\'e}trie
  Diff{\'e}rentielle Cat{\'e}goriques}}
  \bibinfo{volume}{43}(\bibinfo{number}{3}), pp. \bibinfo{pages}{221--239}.

\bibitemdeclare{inproceedings}{power1991n}
\bibitem{power1991n}
\bibinfo{author}{John \surnamestart Power\surnameend} (\bibinfo{year}{1991}):
  \emph{\bibinfo{title}{An $n$-categorical pasting theorem}}.
\newblock In: {\sl \bibinfo{booktitle}{Category theory}},
  \bibinfo{organization}{Springer}, pp. \bibinfo{pages}{326--358},
  \doi{10.1007/BFb0084230}.

\bibitemdeclare{article}{street1976limits}
\bibitem{street1976limits}
\bibinfo{author}{Ross \surnamestart Street\surnameend} (\bibinfo{year}{1976}):
  \emph{\bibinfo{title}{Limits indexed by category-valued 2-functors}}.
\newblock {\sl \bibinfo{journal}{Journal of Pure and Applied Algebra}}
  \bibinfo{volume}{8}(\bibinfo{number}{2}), pp. \bibinfo{pages}{149--181},
  \doi{10.1016/0022-4049(76)90013-X}.

\bibitemdeclare{article}{szpilrajn1930extension}
\bibitem{szpilrajn1930extension}
\bibinfo{author}{Edward \surnamestart Szpilrajn\surnameend}
  (\bibinfo{year}{1930}): \emph{\bibinfo{title}{Sur l'extension de l'ordre
  partiel}}.
\newblock {\sl \bibinfo{journal}{Fundamenta mathematicae}}
  \bibinfo{volume}{16}(\bibinfo{number}{1}), pp. \bibinfo{pages}{386--389}.

\bibitemdeclare{book}{thue1914probleme}
\bibitem{thue1914probleme}
\bibinfo{author}{Axel \surnamestart Thue\surnameend} (\bibinfo{year}{1914}):
  \emph{\bibinfo{title}{Probleme {\"u}ber Ver{\"a}nderungen von Zeichenreihen
  nach gegebenen Regeln}}.
\newblock \bibinfo{series}{Skrifter utg. av Videnskapsselsk. i Kristiania. 1.
  Matem.-naturv. Klasse}.

\bibitemdeclare{inproceedings}{winskel1995models}
\bibitem{winskel1995models}
\bibinfo{author}{Glynn \surnamestart Winskel\surnameend} \&
  \bibinfo{author}{Mogens \surnamestart Nielsen\surnameend}
  (\bibinfo{year}{1995}): \emph{\bibinfo{title}{Models for concurrency}}.
\newblock In: {\sl \bibinfo{booktitle}{Handbook of logic in computer science
  (vol. 4)}}, \bibinfo{organization}{Oxford University Press}, pp.
  \bibinfo{pages}{1--148}.

\end{thebibliography}
\end{document}